\documentclass[12pt]{amsart}
\usepackage{amssymb,amsthm,amsmath}
\usepackage[numbers,sort&compress]{natbib}
\usepackage{color}
\usepackage{graphicx}
\usepackage{tikz}
\usepackage{amssymb,amsthm,amsmath}
\usepackage[numbers,sort&compress]{natbib}
\usepackage{color}
\usepackage{graphicx}
\usepackage{tikz}
\usepackage{xparse}
\title[ Logarithmic bounds of   moments for long  range operators]{	Power law logarithmic bounds of   moments for long  range operators  in arbitrary dimension}

\author{Wencai Liu}
\address[W. Liu]{ Department of Mathematics, Texas A\&M University, College Station, TX 77843-3368, USA} \email{liuwencai1226@gmail.com; wencail@tamu.edu}

\keywords{Green's function, sublinear bounds,  long range operators, large deviation theorem, discrepancy, quantum  dynamics.}
\subjclass[2020]{  81Q10 (primary); 35Q41, 35J10 (secondary)}

\theoremstyle{plain}
\newtheorem{theorem}{Theorem}[section]
\newtheorem{corollary}[theorem]{Corollary}
\newtheorem{lemma}[theorem]{Lemma}

\newcommand{\C}{\mathbb{C}}
\newcommand{\R}{\mathbb{R}}
\newcommand{\T}{\mathbb{T}}
\newcommand{\Z}{\mathbb{Z}}
\newcommand{\N}{\mathbb{N}}
\theoremstyle{definition}
\newtheorem{definition}[theorem]{Definition}

\newtheorem{remark}{Remark}

\begin{document}
	
	% \newcommand{\N}{\mathbb{N}}
	%%% ----------------------------------------------------------------------

	%%% ----------------------------------------------------------------------
	\maketitle
	%%% ----------------------------------------------------------------------
	%\tableofcontents
\begin{center}
 {\it Dedicated  to  Abel Klein  on the occasion of his 75th birthday}
\end{center}
%	\ \ \ \ {\itshape Dedicated  to  Abel Klein  on the occasion of his 75th birthday
	\begin{abstract}
	We show that the sublinear bound of the  bad Green's functions implies explicit logarithmic  bounds  of  moments for   long range  operators  in arbitrary dimension.
	\end{abstract}

	\section{Introduction}

 	In this paper, we are interested in the quantum dynamics of   long range operators on the lattice $\Z^d$.
	 For a self-adjoint operator
	$H$ on $\ell^2(\Z^d), $ $\phi\in \ell^2(\Z^d)$ and $p>0$, let $ \langle|X_H|_{\phi}^p\rangle(t)$ be  the  $p$th moment of  the position operator
 \begin{equation}\label{mom1}
 \langle|X_H|_{\phi}^p\rangle(t)=\sum_{n\in\Z^d}|n|^p|(e^{- itH}\phi,\delta_n)|^2 ,
 \end{equation}
	and $	\langle|\tilde{X}_H|_{\phi}^p\rangle(T)$ be   the time-averaged $p$th moment of  the position operator
	\begin{equation}\label{mom}
	\langle|\tilde{X}_H|_{\phi}^p\rangle(T)=\frac{2}{T}\int_0^\infty e^{-2t/T}\sum_{n\in\Z^d}|n|^p|(e^{- itH}\phi,\delta_n)|^2 dt.
	\end{equation}

	The   moments   $\langle|\tilde{X}_H|_{\phi}^p\rangle(T)$ and $\langle|X_H|_{\phi}^p\rangle(t)$ characterize  how
	fast does $e^{- itH}\phi$ spread out, which are  closely related to the spectral measure $\mu_{\phi}$.
	For example, dynamical localization, namely for any $\phi$, $\langle|X_H|_{\phi}^p\rangle(t)$ is uniformly bounded, implies  only  pure point spectrum of $H$ ~\cite{CFKS87} and continuity (with respect to  the Hausdorff measure) of the spectral measure $\mu_{\phi}$  leads to a power law lower bound of  $\langle|\tilde{X}|_{\phi}^p\rangle(T)$ (see ~\cite{l96} and references therein).
%	Therefore, it is interesting to study the growth or/and decaying property of $\langle|X|_{\phi}^p\rangle(T)$.

	For 	Anderson model with large disorder or the spectral edges (for  one dimensional case, it   holds for the full spectrum), Anderson/dynamical localization holds \cite{jxz19,df,ckm87,am,a94,gd,ds01,fmss,fs,dk89,gk01,gmp,ks80,ek19,ek20,gk13,gk03,bk05,ds20,lz22,kt18,ek16,kn15,kn13}.  For other random models, localization/delocalization  has been extensively studied  as well \cite{eks18,eks182,ek14,ks12,ghk,gks,gks09,aenss}. 
	It has been conjectured that Anderson model in any dimension $d\geq 3$ has the localization (pure point spectrum)-the extended state (absolutely continuous spectrum) transition, which   is still quite open. The transitions of moments  are usually easier to study ~\cite{knr,ks11,gk04,r12,js07,jss}.  %Among all the work, Abel Klein 
For the random operators,  the work of Abel Klein has been crucial in establishing the current state of the art.

	Unlike random operators,  spectral types (singular continuous spectrum or pure point spectrum) of one dimensional quasi-periodic Schr\"odinger operators in the positive Lyapunov exponent regime depend on the arithmetics  of frequencies and phases~\cite{js94,as82,gordon1976,jl118,jl222,ayz,sar}.   The  estimates on  moments are often more stable, namely less sensitive to arithmetics  of frequencies and phases.  Study in this line expects  to obtain asymptotics of moments for all phases under suitable conditions on frequencies ~\cite{HJ19,JP22,DK07,DKJFA,SS21,JM17} (restrictions on frequencies are necessary due to the work of Jitomirskaya and Zhang ~\cite{jz22}). 
		In this paper,  we will study upper bounds of   moments of long range operators, particularly when $H$  is a quasi-periodic  operator (defined on the lattice $\Z^d$ and driven by  base dynamics on the torus $\T^b$).
	
	For one-dimensional ($d=1$) quasi-periodic Schr\"odingers with the shift dynamics on the torus $\T$,  the celebrated work of Damanik-Tcheremchantsev  implies that  when the Lyapunov exponents are positive (in their setting, the potential is a trigonometric function), the moments     have sub-polynomial  growth in time~\cite{DK07,DKJFA}. 
 Jitomirskaya and Mavi  ~\cite{JM17} improved  their results to rough potentials. 
	Han and Jitomirskaya ~\cite{HJ19} generalized earlier works to the torus in arbitrary  dimension with more general base dynamics.
		In ~\cite{JP22}, Jitomirskaya and Powell combined techniques of Damanik-Tcheremchantsev~\cite{DK07} with estimates developed from the proof of Anderson localization for quasi-periodic Schr\"odinger operators  and obtained a  logarithmic bound of moments  (related to  an earlier work of Landrigan-Powell about logarithmic dimensions of spectral measures). 
	In all their work, the transfer matrices and an idea from Jitomirskaya-Last    developed to study spectral dimensions~\cite{JL20} play crucial roles. 
	
	Recently, Jitomirskaya and the author~\cite{JL21}  introduced a new approach-Green's function estimates and show the sub-polynomial growth of moments (by modification, their arguments could lead to logarithmic bounds),  which works for long range operators as well.  

	All the aforementioned   work   focuses on the lattice $\Z$. Very recently, Shamis and Sodin~\cite{SS21} developed an   approach to establish power law logarithmic bounds of moments based on large deviation estimates of Green's functions in arbitrary dimension.  Using the  author's earlier results  on large deviation theorems of Green's functions ~\cite{liuapde}, they obtained power law logarithmic bounds  for various  long range operators on the lattice $\Z^d$.

	In the present work, we introduce a different  approach to  study the 	power law logarithmic bounds of moments. 
	Our idea is inspired by  the localization proof of quasi-periodic Schr\"odinger operators developed by Bourgain and his collaborators \cite{bg,BGS,bbook,bk19,bou07}  (see a recent survey~\cite{Sc22} for more details) and a generalization by the author~\cite{liuapde}.  Research  in this direction starts with the work of  Bourgain and Goldstein ~\cite{bg}, where the authors studied the one-dimensional quasi-periodic Schr\"odinger operators.  With the development in the subsequent work by Bourgain, Goldstein and Schlag~\cite{bg,BGS,bbook,bou07,bgs01,GS01}, it becomes a robust approach to study many spectral problems of quasi-periodic 
	operators. Recently, Jitomirskaya, Shi and the author~\cite{jls20} extended  (also streamlined Bourgain's proof) Bourgain’s results in \cite{bou07} to arbitrary dimension of   frequencies, and the author proved a quantitative and non-self-adjoint version of the work in  \cite{BGS} in arbitrary lattice dimension~\cite{liuapde}. Thanks to all previous works such as Cartan's estimates and techniques from semi-algebraic sets, the localization proof boils down to establish a  sublinear bound (a discrepancy problem) of the bad Green's functions.  Our main result   shows that the sublinear bound immediately implies   explicit power law logarithmic bounds of moments.

	Finally,
we want to compare  logarithmic bounds obtained in this paper, and by Jitomirskaya-Powell~\cite{JP22} and  Shamis-Sodin~\cite{SS21}. Assume that the sublinear bound is $N^{1-\delta}$ (see \eqref{gsublinear} below for the precise definition).
Roughly speaking (see Corollary \ref{cor1}), our main theorem says  that the rate of power law logarithmic bounds  is $\frac{1}{\delta}$.
 The rate of power law logarithmic bounds  in \cite{SS21} by Shamis-Sodin comes from  large deviation estimates of  Green's functions, which is usually bigger than ours because of the     extra dimension loss (see Remark 10 in \cite{liuapde} for the details of dimension loss). See  item 2 in  Remark \ref{re3} and item 2 in Remark \ref{re5}.
We should point out that results in \cite{JP22} only work for Schr\"odinger operators (not long range operators) on the lattice  $\Z$. 
%For Schr\"odinger operators  on  the lattice $\Z$, our results are better than those in \cite{JP22}.  
As mentioned in  \cite[Remark 6]{JP22}, they   obtain the (implicit) rate $\frac{C(b)}{\delta}$, where $C(b)$ is a constant depending on the dimension of the torus. Our rate shows that $C(b)$ can be  1.

%	When 
%	The power law bounds for
%	$\langle|X|_{\phi}^p\rangle(T)$ are characterized by the
%	following upper transport  exponents  $\beta_{\phi}^+(p)$:
%	\begin{equation}\label{dyn-exp}
%	\beta_{\phi}^+(p)=\limsup_{t\to\infty}\frac{\ln\langle|X|_{\phi}^p\rangle(t)}{p\ln t}.
%	\end{equation}
%	Define 	logarithmic transport exponents as
%	\begin{equation}\label{qddyn-exp}
%\tilde{	\beta}_{{\rm ln},\phi}^+(H,p)=\limsup_{T\to\infty}\frac{\ln\langle|\tilde{X}|_{\phi}^p\rangle(T)}{p\ln \ln T},
%	\end{equation}
%	and 
%	\begin{equation}\label{qddyn-exp1}
%	\beta_{{\rm ln},\phi}^+(H,p)=\limsup_{t\to\infty}\frac{\ln\langle|X|_{\phi}^p\rangle(t)}{p\ln \ln t}.
%	\end{equation}
%	

	\section{Main results}

	Let $H$ be a  long range operator  acting on $u=\{u_n\}_{n\in\Z^d}$ in the following form:
	\begin{equation}\label{g01}
	(Hu)_n = \left( \sum_{n'\in \Z^d}H(n,n')u_{n'}\right).
	\end{equation}
	Assume that $H$
	satisfies
	\begin{enumerate}
		\item[a.]
		for any $n,n'\in\Z^d$,
			\begin{equation}\label{GO}
		|H(n,n')|\leq   C_1e^{-c_1|n-n'|},  C_1>0, c_1>0,
		\end{equation}
		where $|n|:=\max\limits_{1\leq i\leq d}|n_i|$ for $n=(n_1,n_2,\cdots,n_d)\in \Z^d$;
			\item[b.]  
		for any  $n,n'\in\Z^d$, 
		\begin{equation}\label{GO1}
		H(n,n')=\overline{H(n',n)}.
		\end{equation}
	\end{enumerate}

	%We remark that all results in this paper work for any  super fast decaying initial vector. 	
%	Here we study %periodic finite-range 

%	where  $V=\{V_n\}_{n\in\Z^d}$ is  real bounded (called potential). 
	%For simplicity,  denote by 
%$	H=A+V$.

	For $d=1$,   the elementary region of size $N$ centered at 0 is given by
	\begin{equation*}
	Q_N=[-N,N].
	\end{equation*}
	
	For $d\geq 2$, denote by $Q_N$ an elementary region of size $N$ centered at 0, which is one of the following regions,
	\begin{equation*}
	Q_N=[-N,N]^d
	\end{equation*}
	or
	$$Q_N=[-N,N]^d\setminus\{n\in\mathbb{Z}^d: \ n_i\varsigma_i 0, 1\leq i\leq d\},$$
	where  for $ i=1,2,\cdots,d$, $ \varsigma_i\in \{<,>,\emptyset\}$ and at least two $ \varsigma_i$  are not $\emptyset$.

	Denote by $\mathcal{E}_N^{0}$ the set of all elementary regions of size $N$ centered at 0. 
	%Let $\mathcal{E}_N$ be the set of all translates of  elementary regions  with center at 0, namely,
	Let
	$$\mathcal{E}_N:=\{n+Q_N:n\in\mathbb{Z}^d,Q_N\in \mathcal{E}_N^{0}\}.$$
	We call  elements in $\mathcal{E}_N$ elementary regions.

	Let  $R_{{\Lambda}}$  be the operator of  restriction to $\Lambda \subset \Z^d$. Define the Green's function by
	\begin{equation}\label{g0}
	G_{{\Lambda}}(z)=(R_{{\Lambda}}(H-zI)R_{{\Lambda}})^{-1}.
	\end{equation}
	Set $ G(z)=(H-zI)^{-1}$.  Clearly, both $  G_{{\Lambda}}(z)$ and $G(z)$ are always well defined for $z\in \C_+\equiv \{z\in \C: \Im z>0\}$. Sometimes, we drop the dependence on $z$  for simplicity.
	
	We say an elementary region $ \Lambda\in \mathcal{E}_{N}$ is  in class {\it G} (Good)  if
\begin{equation}\label{ggood}
| G_{\Lambda}(n,n^\prime)|\leq  e^{-c_2|n-n^\prime|}, \text{ for } |n-n^\prime|\geq \frac{N}{10},
\end{equation}
where $0<c_2\leq  c_1$.

	Since the self-adjoint operator $H$ given by \eqref{g01} is bounded, there  exists a large 
	$K>0$  such that $\sigma({H})\subset [-K+1,K-1]$. 

For $\Lambda\subset \Z^d$, 
	denote by $\partial \Lambda$ its boundary.
		\begin{definition}
  Fix $ \varsigma\in(0,1)$.
		Let $ {\Lambda}_0\in \mathcal{E}_N$ be an elementary region.   Given $\xi$ with $ 0<\xi<1$,  we say ${\Lambda}_0$ satisfies sublinear bound property with the parameter $\xi $ if  
	for any family $\mathcal{F}$ of pairwise disjoint elementary regions in ${\Lambda}_0$ with size $M=\lfloor N^\xi \rfloor $,
		\begin{equation}\label{qdgnbad}
		\# \{ \Lambda \in \mathcal{F}: \Lambda \text { is not in class G } \}\leq \frac{N^{\varsigma}}{N^\xi}.
		\end{equation}
	\end{definition}

	\begin{theorem}\label{thm1}
		 
	Suppose there exist    $\epsilon_0>0$ and $N_0>0$ such that the following is true.
	Let  $z=E+i\epsilon$ with $|E|\leq K$ and $0<\epsilon\leq \epsilon_0$.
For any $n\in\Z^d$ with $|n|\geq N_0$ there exists  ${\Lambda}_0\in \mathcal{E}_{ N}$ such that $\frac{|n|}{100}\leq N \leq \frac{|n|}{10}$, $n\in {\Lambda}_0$, ${\rm dist } (n,\partial {\Lambda}_0 )\geq \frac{N}{5}$ and 
 ${\Lambda}_0$  satisfies sublinear bound property with the parameter $\xi $.
	Then  for any $\phi$ with compact support and any $\varepsilon>0$ there exists $T_0>0$ (depending on $d,p,\phi,K$, $\varsigma,\xi, \epsilon_0$, $c_1$, $c_2$, $C_1$, $N_0$ and $\varepsilon$) such that for any $T\geq T_0$
		\begin{equation} \label{gdq1}
	\langle|\tilde{X}_H|_{\phi}^p\rangle(T)	\leq  (\ln T)^{  \frac{p}{\xi} +\varepsilon}
	\end{equation}
	and  for any $t\geq T_0$
	\begin{equation} \label{gdq2}
	\langle|{X}_H|_{\phi}^p\rangle(t)	\leq  (\ln t)^{  \frac{p}{\xi} +\varepsilon}.
\end{equation}

	%where $\vartheta=\vartheta(\sigma,\tilde{\sigma},\xi,\varsigma)>0$.
\end{theorem}

Fixed $0<\sigma<1$,
	we say  an elementary region $ \Lambda\in \mathcal{E}_{N}$ is  in class {\it SG$_{N}$} (strongly good with size $N$) if
\begin{equation}\label{ggoodt1}
||G_{\Lambda}||\leq  e^{N^{\sigma}},
\end{equation}
and
\begin{equation}\label{ggoodt2}
|G_{\Lambda}(n,n^\prime)|\leq  e^{-c_2|n-n^\prime|}, \text{ for } |n-n^\prime|\geq \frac{N}{10},
\end{equation}
where $0<c_2\leq c_1$. %When there is no confusion, we drop  the dependence of $N$ from  the notation {\it SG$_{N}$}.
\begin{corollary}\label{cor1}
Define $\mathcal{B}_{N,N_1}$ as
	\begin{equation*}
	\mathcal{B}_{N,N_1}=\{ n\in [-N,N]^d: \text{ there exists } Q_{N_1}\in  \mathcal{E}_{N_1}^{0} \text{ such that } n+Q_{N_1}\notin \text{{\it SG}}_{N_1}\}
	\end{equation*}
	Assume that there exists $\epsilon_0>0$ such that for any $z=E+i\epsilon $ with  $|E|\leq K$ and $0<\epsilon\leq \epsilon_0$, and arbitrarily  small $\varepsilon>0$,
	\begin{equation}\label{gsublinear}
 \#  \mathcal{B}_{N,\lfloor N^{\varepsilon} \rfloor}\leq N^{1-\delta} {\text{ when }} N \geq N_0 
	\end{equation}
	($N_0$ may depend on $\varepsilon$).
	Then  for any $\phi$ with compact support and any $\varepsilon>0$ there exists $T_0>0$ (depending on $d,p,\phi,K$, $\sigma,\delta, \epsilon_0$, $c_1$, $c_2$, $C_1$, $N_0$ and $\varepsilon$) such that for any $T\geq T_0$
	\begin{equation} \label{gdq11}
	\langle|\tilde{X}_H|_{\phi}^p\rangle(T)	\leq  (\ln T)^{  \frac{p}{\delta} +\varepsilon}
	\end{equation}
	and  for any $t\geq T_0$
	\begin{equation} \label{gdq12}
	\langle|{X}_H|_{\phi}^p\rangle(t)	\leq  (\ln t)^{  \frac{p}{\delta} +\varepsilon}.
	\end{equation}

\end{corollary}
\begin{remark}
In applications,  $\varsigma=1-\varepsilon$ with arbitrarily  small $\varepsilon>0$. Then the upper bound  in \eqref{qdgnbad}   equals   $N^{1-\xi-\varepsilon}$.  Both $N^{1-\xi-\varepsilon}$ and $N^{1-\delta}$ in \eqref{gsublinear} are referred to as the sublinear  bound property of (bad) Green's functions. 
\end{remark}

%\begin{theorem}\label{thm3}
%	Let $d=1$.
%	Suppose there exist    $\epsilon_0>0$ and $N_0>0$ such that the following is true.
%	Let  $z=E+i\epsilon$ with $|E|\leq K$ and $0<\epsilon\leq \epsilon_0$.
%	Suppose for    $N>N_0$, there exists an  interval $I\subset [-\frac{ N}{2}, -\frac{ N}{4}]$ or $I\subset [\frac{ N}{4}, \frac{ N}{2}]$ such that
%	$|I|  \geq  N^{\epsilon_1}  $
%	and  for any $n,n^\prime\in {I}$ and $|n-n^\prime|\geq \frac{1}{20}|I|$,  we have
%	\begin{equation*}
%	| G_{{I}}(z)(n,n^\prime)|\leq e^{-c|I|}.
%	\end{equation*}
%	Then  for any $\phi$ with compact support, 
%	\begin{equation} 
%	\beta_{{\rm ln},\phi}^+(p)\leq \frac{1}{\epsilon_1}.
%	\end{equation}
%\end{theorem}

\section{Proof of Theorem \ref{thm1}  and Corollary \ref{cor1}}
	
 Let us first recall some notations from \cite{liuapde}.

The  width  of a   subset  $\Lambda\subset \Z^d$, is defined by maximum
$M\in \N$  such that  for any  $n\in \Lambda$, there exists  $\hat{M}\in \mathcal{E}_M$ such that
\begin{equation*}
n\in \hat{M} \subset \Lambda
\end{equation*}
and
\begin{equation*}
\text{ dist }(n,\Lambda\backslash \hat{M})\geq M/2.
\end{equation*}

A generalized  elementary region is defined to be a subset $\Lambda\subset \Z^d$ of the form
\begin{equation*}
\Lambda:= R\backslash(R+y),
\end{equation*}
where $y\in\Z^d$ is arbitrary and $R$ is a rectangle,
\begin{equation*}
R=\{n=(n_1,n_2,\cdots,n_d)\in \Z^d: |n_1-n_1^\prime|\leq M_1, \cdots,|n_d-n_d^\prime|\leq M_d\}.
\end{equation*}

For $ \Lambda\subset\mathbb{Z}^d$,   denote by 
$\mathrm{diam}(\Lambda)=\sup_{n,n'\in \Lambda}|n-n'|$ its diameter.

Denote by $\mathcal{R}_N$
all  generalized elementary regions with diameter less than or equal  to $N$.
Denote by $\mathcal{R}_N^M$
all   generalized elementary regions in $\mathcal{R}_N$ with width larger than or equal to $M$.

Let us collect and define some notations which will be used throughout the proof: 
 \begin{itemize}
 	\item $\sigma(H)\subset [-K+1,K-1]$. 
 	\item $z=E+ \epsilon i$, $|E|\leq K$ and $0<\epsilon\leq \epsilon_0$. 
 	\item $\epsilon=\frac{1}{T}$.
 	\item $G_{\Lambda}= G_{\Lambda}(z)=  (R_{\Lambda} (H-z) R_{\Lambda})^{-1}= (R_{\Lambda} (H-E-\epsilon i) R_{\Lambda})^{-1}$ 
 	\item  $G= (H-zI)^{-1}$.
 	\item $\text {supp } \phi \subset [-K_1,K_1]^d$
 	\item $C(c)$ is a large (small) constant.
 	
 \end{itemize}

%We  remark that the upper bound $\frac{5^{\tilde{\sigma}}-1}{5^{\tilde{\sigma}}}c_1$ is chosen for technical convenience.  
	\begin{theorem}\label{thmmul}
%		Assume $A$ satisfies \eqref{GO}.
		%Suppose that $M$, $N$ are positive integers such that for some $0<\tau<1$
		%\begin{equation*}
		% N^{\tau}\leq M\leq 2N^{\tau}.
		%\end{equation*}
		Let $\varsigma,\sigma,\xi\in(0,1)$.
		Let ${\Lambda}_0\in \mathcal{E}_N$ be an elementary region with the property that for all $ \Lambda\subset {\Lambda}_0$, $\Lambda\in \mathcal{R}_L^{N^\xi}$ with $N^{\xi}\leq L\leq 2N$, the Green's function $G_{\Lambda}$ satisfies 
		\begin{equation}\label{gboundnov26}
		||G_{\Lambda}||\leq e^{L^{\sigma}}.
		\end{equation}
		Assume that $c_2\leq \frac{4}{5}c_1$ and for any family $\mathcal{F}$ of pairwise disjoint elementary regions in ${\Lambda}_0$ with size $M=\lfloor N^\xi \rfloor $,
		\begin{equation}\label{gnbad}
		\# \{ \Lambda \in \mathcal{F}: \Lambda \text { is not in class G } \}\leq \frac{N^{\varsigma}}{N^\xi}.
		\end{equation}
		Then
		for large $N$ (depending on $C_1,c_1,\varsigma,\sigma,\xi$ and the lower bound of $c_2$),
		\begin{equation}\label{gthm1}
		|G_{\Lambda_0}(n,n^\prime)|\leq  e^{-(c_2-N^{-\vartheta})|n-n^\prime|},\text{ for } |n-n^\prime|\geq \frac{N}{10},
		\end{equation}
		where $\vartheta=\vartheta(\sigma,\xi,\varsigma)>0$.
	\end{theorem}
\begin{remark}
	 Theorem \ref{thmmul}  in the settings   of  Schr\"odinger operators ($\Delta+V$) on $\Z^2$  was   proved in \cite{BGS}.
	 The author generalized their proof to the settings in  Theorem \ref{thmmul} in \cite{liuapde}.
\end{remark}
	Assume $\Lambda_1$ and $\Lambda_2$ are two disjoint subsets of  $\Z^d$. 
%Namely,  $\Lambda_1,\Lambda_2 \subset \Z^d$ and $\Lambda_1\cap\Lambda_2=\emptyset$. 
Let $\Lambda=\Lambda_1\cup \Lambda_2$.
Suppose that $ R_{\Lambda}AR_{\Lambda}$ and $ R_{\Lambda_i}AR_{\Lambda_i}$, $i=1,2$ are invertible.
Then
\begin{equation*}
G_{\Lambda}=G_{\Lambda_1}+G_{\Lambda_2}-(G_{\Lambda_1}+G_{\Lambda_2})( H_{\Lambda}-H_{\Lambda_1}-H_{\Lambda_2})G_{\Lambda}.
\end{equation*}
%If $m\in \Lambda_1$ and $n\in \Lambda$, we have
%\begin{equation}\label{Greso}
%|G_{\Lambda}(m,n)|\leq |G_{\Lambda_1}(m,n)|\chi_{\Lambda_1}(n)+ \sum_{n^{\prime}\in \Lambda_1,n^{\prime\prime}\in \Lambda_2} e^{-c_1|n^{\prime}-n^{\prime\prime}|}|G_{\Lambda_1}(m,n^{\prime})||G_{\Lambda}(n^{\prime\prime},n)|.
%\end{equation}
If $n\in \Lambda_2$ and $m\in \Lambda$, we have
\begin{equation}\label{Greson}
|G_{\Lambda}(m,n)|\leq |G_{\Lambda_2}(m,n)|\chi_{\Lambda_2}(n)+ \sum_{n^{\prime}\in \Lambda_1,n^{\prime\prime}\in \Lambda_2} e^{-c_1|n^{\prime}-n^{\prime\prime}|}|G_{\Lambda}(m,n^{\prime})||G_{\Lambda_2}(n^{\prime\prime},n)|.
\end{equation}

\begin{lemma}\label{key1}
	Fixed any $\sigma\in (0,1)$,
	let $N\geq  (\log \frac{1}{\epsilon})^{\frac{1}{\xi \sigma}}$.
	Assume that 
${\Lambda}_0\in \mathcal{E}_{ N}$,  $n\in {\Lambda}_0$, ${\rm dist } (n,\partial {\Lambda}_0 )\geq \frac{N}{5}$ and 
${\Lambda}_0$  satisfies the sublinear bound property with the parameter $\xi $.  Then for any $j$ with $|j|\leq K_1$,
\begin{equation}
|((H-E-\epsilon i)^{-1}\delta_j,\delta_n)| \leq C\epsilon^{-2} e^{-c|n|}
\end{equation}
\end{lemma}
\begin{proof}
	Since $H$ is self-adjoint, one  has that   for any $\Lambda \subset \Z^d$,  ${\rm dist} (\sigma(H_{\Lambda}), z)\geq \epsilon$ and hence 
	\begin{equation*}
||G_{\Lambda}||\leq \epsilon^{-1}. 
	\end{equation*}
Then for any $ \Lambda\subset {\Lambda}_0$, $\Lambda\in \mathcal{R}_L^{N^\xi}$ with $N^{\xi}\leq L\leq 2N$, one has that  the Green's function $G_{\Lambda}$ satisfies 
\begin{equation}\label{gdgboundnov26}
||G_{\Lambda}||\leq \epsilon^{-1}\leq e^{L^{\sigma}},
\end{equation}
where the second inequality in \eqref{gdgboundnov26} holds by the assumption.
By Theorem \ref{thmmul}, one has that 
	\begin{equation}\label{gdgthm1}
|G_{\Lambda_0}(n,n^\prime)|\leq  e^{-c|n-n^\prime|},\text{ for } |n-n^\prime|\geq \frac{N}{10}.
\end{equation}
By \eqref{Greson} (applying $\Lambda_2=\Lambda_0$ and $\Lambda=\Z^d$) and using that $j\notin \Lambda_0$, one has that 
\begin{equation}\label{qdg14}
| G(j,n)|\leq  C\sum_{n^{\prime}\in \Lambda\backslash \Lambda_0,n^{\prime\prime}\in \Lambda_0} e^{-c_1|n^{\prime}-n^{\prime\prime}|}|G(j,n^{\prime})||G_{\Lambda_0}(n^{\prime\prime},n)|.
\end{equation}
If  ${\rm dist } (n^{\prime\prime},\partial {\Lambda}_0 )\geq \frac{N}{1000}$,  then $|n'-n^{\prime\prime}|\geq \frac{N}{1000}$ and hence $e^{-c_1|n'-n^{\prime\prime}|}\leq e^{-c N}$.

If  ${\rm dist } (n^{\prime\prime},\partial {\Lambda}_0 )\leq \frac{N}{1000}$, then $|n-n^{\prime\prime}|\geq \frac{N}{6}$. By \eqref{gdgthm1}, one has that
$|G_{\Lambda_0}(n^{\prime\prime},n)|\leq e^{-cN}$. 

Therefore,  one concludes  that 
\begin{align}
|((H-E&-\epsilon i)^{-1}\delta_j,\delta_n)|\nonumber \\
 \leq  &C\sum_{n^{\prime}\in \Lambda\backslash \Lambda_0,n^{\prime\prime}\in \Lambda_0\atop{{\rm dist } (n^{\prime\prime},\partial {\Lambda}_0 )\geq \frac{N}{1000}}} e^{-c_1|n^{\prime}-n^{\prime\prime}|}|G(j,n^{\prime})||G_{\Lambda_0}(n^{\prime\prime},n)|\nonumber\\
  &+ C\sum_{n^{\prime}\in \Lambda\backslash \Lambda_0,n^{\prime\prime}\in \Lambda_0\atop{{\rm dist } (n^{\prime\prime},\partial {\Lambda}_0 )< \frac{N}{1000}}} e^{-c_1|n^{\prime}-n^{\prime\prime}|}|G(j,n^{\prime})||G_{\Lambda_0}(n^{\prime\prime},n)|\nonumber\\
 \leq  & C\epsilon^{-2}\sum_{n^{\prime}\in \Lambda\backslash \Lambda_0,n^{\prime\prime}\in \Lambda_0\atop{|n'-n^{\prime\prime}|\geq \frac{N}{1000}} } e^{-c|n'-n^{\prime\prime}|}+ C\epsilon^{-1}e^{-cN}\sum_{n^{\prime}\in \Lambda\backslash \Lambda_0,n^{\prime\prime}\in \Lambda_0} e^{-c|n'-n^{\prime\prime}|}\nonumber\\
\leq &C\epsilon^{-2} e^{-cN}\leq C\epsilon^{-2}e^{-c|n|}.
\end{align}
 
\end{proof}
The following lemma follows from  Lemma 2 in \cite{DKJFA}. 
\begin{lemma}
	Assume that $\phi$ has compact  support. Then for any  large $|n|$,
	\begin{equation}\label{Part11}
 |(e^{- itH}\phi,\delta_n)|^2 \leq e^{-c|n|}+\frac{1}{t} \int_{-K}^{K}|  ((H-E-\frac{i}{t})^{-1}\phi,\delta_n)|^2 d E
	\end{equation}
\end{lemma}
\begin{proof}[\bf Proof of Theorem \ref{thm1}]
Fix any $\sigma$ in $(0,1)$. 
For any $j$ with $|j|\leq K_1$,
let
\begin{equation}\label{g-1}
a(j,n,T)=\frac{2}{T}\int_{0}^{\infty}e^{-2t/T}|(e^{- itH}\delta_j,\delta_n)|^2 dt,
\end{equation}
and then
\begin{eqnarray}\label{g20}
\langle|\tilde{X}_H|_{\phi}^p\rangle(T) &\leq & C \sum_{|j|\leq K_1}\sum_{n\in\Z^d}|n|^pa(j,n,T) 
\end{eqnarray}
By the Parseval formula
\begin{equation}\label{Par}
a(j,n,T)=\frac{1}{T\pi} \int_{-\infty}^{\infty}|((H-E-\frac{i}{T})^{-1}\delta_j,\delta_n)|^2 d E.
\end{equation}
Recall that $\sigma(H)\subset [-K+1,K-1]$. 
For any $E\in(-\infty,-K)\cup  (K,\infty)$,   $\text{dist} (E+\frac{i}{T},\text{spec}(H))\geq 1$.
The  well-known Combes-Thomas  estimate (e.g. A.11 in \cite{GKT04}) yields 
\begin{equation}\label{CT}
|((H-E-\frac{i}{T})^{-1}\delta_j,\delta_n)|\leq Ce^{-c|n|}.
\end{equation}
By \eqref{Par} and \eqref{CT},  one has that 
\begin{equation}\label{Par1}
a(j,n,T)\leq Ce^{-c|n|}+\frac{1}{T\pi} \int_{-K}^{K}|((H-E-\frac{i}{T})^{-1}\delta_j,\delta_n)|^2 d E.
\end{equation}
Rewrite \eqref{g20}:
\begin{eqnarray}
\langle|\tilde{X}_H|_{\phi}^p\rangle(T) &\leq & C \sum_{|j|\leq K_1}\sum_{n\in\Z^d}|n|^pa(j,n,T) \nonumber\\
&\leq &C\sum_{|j|\leq K_1}\sum_{n\in\Z^d\atop^{ |n|\geq  (\log T)^{\frac{1}{\xi \sigma}}}} |n|^pa(j,n,T)  +C\sum_{|j|\leq K_1}\sum_{n\in\Z^d\atop^{|n|\leq  (\log T)^{\frac{1}{\xi \sigma}}}} |n|^pa(j,n,T).  \label{g13}
\end{eqnarray}
By Lemma \ref{key1} and \eqref{Par1}, one has that 
\begin{align}
\sum_{n\in\Z^d\atop^{ |n|\geq  (\log T)^{\frac{1}{\xi \sigma}}}} |n|^p|a(j,n,T)&\leq \sum_{n\in\Z^d\atop^{ |n|\geq  (\log T)^{\frac{1}{\xi \sigma}}}} |n|^p T^3e^{-c|n|} \nonumber\\
& \leq T^3e^{-c (\log T)^{\frac{1}{\xi\sigma}}}\leq C.\label{g12}
\end{align}
Direct computations imply that 
\begin{align}
C\sum_{n\in\Z^d\atop^{|n|\leq  (\log T)^{\frac{1}{\xi \sigma}}}} |n|^pa(j,n,T)& \leq C(\log T)^{\frac{p}{\xi\sigma}}\sum_{n\in\Z^d} a(j,n,T) \nonumber\\
&= C(\log T)^{\frac{p}{\xi\sigma}},\label{g11}
\end{align}
	where \eqref{g11} holds by the fact that $\sum_{n\in\Z^d} a(j,n,T) =1$.
	
By \eqref{g13}, \eqref{g12} and \eqref{g11}, we conclude that 
\begin{equation}\label{g14}
\langle|\tilde{X}_H|_{\phi}^p\rangle(T) \leq C(\log T)^{\frac{p}{\xi\sigma}}.
\end{equation}
By letting $\sigma\to 1$, we complete the proof of \eqref{gdq1} .

Replacing \eqref{Par} with \eqref{Part11} and repeating the proof of \eqref{g14}, we  have  \eqref{gdq2}. 
\end{proof}
\begin{proof}[\bf Proof Corollary \ref{cor1} ]
	Let $ \xi=\delta-2C(d)\varepsilon$.
Let ${\mathcal{F}}$ be any pairwise disjoint elementary regions in $[-N,N]^d$ with size    $\lfloor N^\xi\rfloor $.
By \eqref{gsublinear}, one has that  ($N_1=\lfloor N^{\varepsilon}\rfloor$) there are at most $N_1^{C(d)}N^{1-\delta }=N^{1-\delta +C(d)\varepsilon}$ in ${\mathcal{F}}$ will intersect  elementary regions  not in  $SG_{N_1}$.
By resolvent identity arguments (e.g. \cite[Theorem 6.1]{liuapde}), any elementary region in $[-N,N]^d$ with size $\lfloor N^\xi\rfloor $, without intersecting any non-$SG_{N_1}$ elementary regions,  satisfies \eqref{ggood}.
It implies \eqref{gnbad} is true for $\varsigma=1-\varepsilon$.
Applying Theorem \ref{thm1} and letting $\varepsilon \to 0$,  
 we obtain Corollary \ref{cor1}.
\end{proof}

 \section{Applications}

 	 In this section,  the long range operator $S$ on  $\ell^2(\Z^d)$ satisfies 
 	 	 \begin{enumerate}
 	 	 	\item[a.]
 	 	 	for any $n,n'\in\Z^d$,
 	 	 	\begin{equation}\label{GO11}
 	 	 	|S(n,n')|\leq   C_1e^{-c_1|n-n'|},  C_1>0, c_1>0;
 	 	 	\end{equation}
 	 	 %	where $|n|:=\max\limits_{1\leq i\leq d}|n_i|$ for $n=(n_1,n_2,\cdots,n_d)\in \Z^d$;
 	 	 	\item[b.]  
 	 	 	for any  $n,n'\in\Z^d$, 
 	 	 	\begin{equation}\label{GO12}
 	 	 	S(n,n')=\overline{S(n',n)};
 	 	 	\end{equation}
 	 	 	\item[c.]		for any $n\in \Z^d, n^\prime\in \Z^d$ and $k\in \Z^d,$ 
 	 	 	\begin{equation}\label{gdec72}
 	 	 	S(n+k,n^\prime+k)= S(n,n^\prime).
 	 	 	\end{equation}
 	 	 
 	 	 \end{enumerate}

 Let $f$ be a function from $\Z^d\times\T^b$ to $\T^b$.
 Assume for any $m_1,m_2,\cdots,m_d\in \Z^d$ and  $n_1,n_2,\cdots,n_d\in \Z^d$,
 \begin{equation*}
 f({m_1+n_1},{m_2+n_2},\cdots, {m_d+n_d},x)= f({m_1},{m_2},\cdots, {m_d}, f({n_1},{n_2},\cdots,{n_d},x)).
 \end{equation*}
 Sometimes, we write down $f^n(x)$ for $f(n,x)$ for convenience, where $n\in \Z^d$ and $x\in \T^b$.

 Define a family of  operators $H_x$  on $\ell^2(\Z^d)$:
 \begin{equation}\label{ops}
 H_x=S+v(f(n,x))\delta_{nn^\prime},
 \end{equation}
 where  $v$ is a real  analytic   function on $\T^b$.

 In \cite{liuapde}, the author  obtained the large deviation estimates for Green’s functions of  long range operators  in  arbitrary  dimension, which generalized results in \cite{BGS}. As a result, he proved the large deviation theorem of Green's functions and sublinear bound \eqref{gsublinear} for various models. Therefore, we can apply Corollary \ref{cor1} to obtain logarithmic   bounds for many  operators  studied in \cite{liuapde}. 
 For simplicity, we only  discuss applications of  Corollary \ref{cor1} to several  cases. 
 In  this section,  estimates  on $\delta$ in \eqref{gsublinear} can be found in \cite{liuapde}. 
 %We  collect those information from \cite{liuapde}.   
 For readers' convenience, we  will provide the sketch of the proof. 
 %shew that many models satisfy 

	 \subsection{ Discrepancy:  $d=1$, arbitrary $b$}

	 Let $ {x}_1,x_2.\cdots, {x}_N\in [0,1)^b$ and $\mathcal{S}\subset [0,1)^b$.
	 Let $A(\mathcal{S}; \{{x}_j\}_{j=1}^N)$ be the number of  ${x}_j$ ($1\leq j\leq N$)
	 such that ${x}_j\in \mathcal{S}$.
	Let $ D_N(f)$ be 
	 the discrepancy of the sequence $\{f(n,x)\}_{n=1}^N$:
	 \begin{align}\label{counting}
	 D_N(f)=\sup_{x\in\T^b}\sup_{\mathcal{S}\in \mathcal{C}}\left|\frac{A(\mathcal{S}; \{f(n,x)\}_{n=1}^N)}{N}-\mathrm{Leb}(\mathcal S)\right|,
	 \end{align}
	 where $ \mathcal{C}$ is the family of all  intervals in $[0,1)^b$, namely $\mathcal{S}$ has the form of $$\mathcal{S}=[\varrho_1,\beta_1]\times [\varrho_2,\beta_2]\times \cdots \times [\varrho_b,\beta_b]$$ with $0\leq \varrho_k<\beta_k<1$, $k=1,2,\cdots,b$.

	Let $\zeta,\sigma\in(0,1)$.	
	We say the Green's function  of an operator $H_x$  satisfies property LDT (large deviation theorem) in complexified energies (sometimes just say LDT for short)  if  there exists $\epsilon_0>0$ and $N_0>0$ such that for  any  $N\geq N_0$, there exists a  subset $X_N\subset \mathbb{T}^b$  such that
	\begin{equation}\label{gdec3}
	\mathrm{Leb}(X_N)\leq e^{-{N}^{\zeta}},
	\end{equation}
	and for any $x\notin  X_N \mod \Z^b$ and $Q_N\in \mathcal{E}_N^0$,
	\begin{eqnarray*}
		|| G_{Q_N}(z) ||&\leq& e^{N^{\sigma}} ,\\
		|G_{Q_N}(z)(n,n^\prime)| &\leq& e^{-c_2|n-n^\prime|}, \text{ for } |n-n^\prime|\geq \frac{N}{10},
	\end{eqnarray*}
	where  $z=E+i\epsilon$ with $E\in[-K,K]$ (recall that $\sigma(H_x)\subset(-K+1,K-1)$) and $0<\epsilon\leq \epsilon_0$.

	 \begin{theorem}\label{thm4}
	 	Let $d=1$. Assume that for any   $N\geq N_0$, the discrepancy 
	 \begin{equation}
	 D_N(f)\leq N^{-\delta_1}.
	 \end{equation}
	 Assume that $H_x$ satisfies LDT. 
	Then  for any $\phi$ with compact support and any $\varepsilon>0$ there exists $T_0>0$ (depending on $p,N_0,\phi,S,v,f$ and $\varepsilon$) such that for any $T\geq T_0$
	\begin{equation}  
	\langle|\tilde{X}_{H_x}|_{\phi}^p\rangle(T)	\leq  (\ln T)^{  \frac{p}{\delta_1} +\varepsilon}
	\end{equation}
	and  for any $t\geq T_0$
	\begin{equation}  
	\langle|{X}_{H_x}|_{\phi}^p\rangle(t)	\leq  (\ln t)^{  \frac{p}{\delta_1} +\varepsilon}.
	\end{equation}
	
	 \end{theorem}

\begin{proof}
	By approximating the analytic function with trigonometric polynomials  and using Taylor expansions  and standard perturbation arguments,
	we can  assume that $ X_{N}$ (given by \eqref{gdec3}) is a semi-algebraic set with degree less than $N^C$. 
By \cite[Corollary 9.7]{bbook} (also see \cite[Theorem 8.7]{liuapde}), one has that \eqref{gsublinear} holds for any $\delta<\frac{\delta_1}{b}$. Now  Theorem \ref{thm4} follows from Corollary \ref{cor1}.
\end{proof}

\begin{remark}
	The largeness of $T_0$ in this section  does not depend on $x\in\T^b$.
\end{remark}

	\subsection{Shifts: $d=1$, arbitrary $b$}
	
	 Let  $\alpha=(\alpha_1,\alpha_2,\cdots,\alpha_b)\in [0,1)^b$ and 
	 \begin{equation}\label{gshift}
	f (x)=x+\alpha \mod \Z^b, x\in\T^b. 
	 \end{equation}
	 	We say that $\alpha=(\alpha_1,\alpha_2,\cdots,\alpha_b)$ satisfies Diophantine condition ${\rm DC}(\kappa,\tau)$, if
	 \begin{equation}\label{gdc}
	 ||k\cdot \alpha||\geq \frac{\tau}{|k|^{\kappa}},k\in \Z^{b}\backslash \{(0,0,\cdots,0)\}.
	 \end{equation}
	%The following Lemmas are well known.
	\begin{lemma}\cite{dt97}\label{ledissh}
		Assume $\alpha\in {\rm DC}(\kappa,\tau) $.  Let $f$ be given by \eqref{gshift}. Then
		\begin{equation*}
		D_N(f)\leq C(b,\kappa,\tau)N^{-\frac{1}{\kappa}}(\log N)^2.
		\end{equation*}
	\end{lemma}

	Denote by $\Delta$ the discrete Laplacian on $\ell^2(\Z)$, that is, for $\{u(n)\}\in \ell^2(\Z)$,
	\begin{equation*}
	(\Delta u)_n= u_{n+1}+u_{n-1}
	\end{equation*}

	Let $H_x$ on $\ell^2(\Z)$ be given by
	\begin{equation}\label{opapp1}
	H_x=\Delta+ v(f^n(x))=\Delta+ v(x_1+n\alpha_1,x_2+n\alpha_2,\cdots,x_b+n\alpha_b)\delta_{nn'},
	\end{equation}
	where $n,n^\prime\in \Z$ and $v$ is real analytic on $\T^b$.
	%where $x=(x_1,x_2,\cdots,x_b)$ and $n\alpha=(n\
	%  a_1,n\alpha_2,\cdots,n\alpha_b)$.

	\begin{theorem}\label{thmapp1}
		Let  $\alpha\in {\rm DC}(\kappa,\tau)$. Let  $H_x$ be given by \eqref{opapp1}.  
		Assume the Lyapunov exponent  $L(E)$ is positive for all $E\in\R$. 	Then  for any $\phi$ with compact support and any $\varepsilon>0$ there exists $T_0>0$ (depending on $p,\phi,v,\alpha$ and $\varepsilon$) such that for any $T\geq T_0$
		\begin{equation}  
		\langle|\tilde{X}_{H_x}|_{\phi}^p\rangle(T)	\leq  (\ln T)^{  b\kappa p+\varepsilon}
		\end{equation}
		and  for any $t\geq T_0$
		\begin{equation}  
		\langle|{X}_{H_x}|_{\phi}^p\rangle(t)	\leq  (\ln t)^{  b\kappa p +\varepsilon}.
		\end{equation}
	\end{theorem}
\begin{proof}
Since $L(E)>0$ for any $E\in\R$, by the continuity of Lyapunov exponents (e.g  \cite{bbook}, \cite{boujam}  or \cite{p22}), one has that  there exists $\epsilon_0>0$ such that 
$L(E+i\epsilon)>0$  for any $E\in[-K,K]$ and $0<\epsilon\leq \epsilon_0$. This implies that $H_x$ satisfies LDT (e.g. \cite{bbook}).
Now Theorem \ref{thmapp1} follows from Theorem \ref{thm4} and Lemma \ref{ledissh}.
\end{proof}
\begin{remark}\label{re3}
	\begin{enumerate}
		\item Under a stronger Diophantine condition of frequencies $\alpha$, the modulus of continuity and large deviation theorem  of Lyapunov exponents  were first established  in \cite{GS01} (also see \cite{DKJEMS} for a  recent generalization). 
		\item 	 A  larger  power law logarithmic   bound  $  b^3\kappa^2 p$ was established  by  Shamis-Sodin  \cite{SS21}.
		\item An implicit bound $C(b) \kappa p$ ($C(b)$ is a constant depending on $b$) was obtained in \cite{JP22}.
	\end{enumerate}
	
\end{remark}
	When the dimension of the torus is 1, we can improve Theorem \ref{thmapp1}.
\begin{theorem}\label{thm6}
	Let $d=b=1$.
	Let  $\alpha\in {\rm DC}(\kappa,\tau)$. Let  $H_x$ be given by \eqref{opapp1} (with $b=1$).  
	Assume that the Lyapunov exponent  $L(E)$ is positive for any $E\in\R$. 	Then  for any $\phi$ with compact support and any $\varepsilon>0$ there exists $T_0>0$ (depending on $p,\phi,v,\alpha$ and $\varepsilon$) such that for any $T\geq T_0$
	\begin{equation}  
	\langle|\tilde{X}_{H_x}|_{\phi}^p\rangle(T)	\leq  (\ln T)^{p +\varepsilon}
	\end{equation}
	and  for any $t\geq T_0$
	\begin{equation}  
	\langle|{X}_{H_x}|_{\phi}^p\rangle(t)	\leq  (\ln t)^{p +\varepsilon}.
	\end{equation}
\end{theorem}
\begin{proof}
	Similar to the proof of Theorem \ref{thmapp1},  $H_x$ satisfies LDT.  Similar to the proof of Theorem \ref{thm4}, 
	we can  assume that $ X_{N}$ (given by \eqref{gdec3}) is a semi-algebraic set with degree less than $N^C$. 
	Then  \eqref{gsublinear} holds for any $\delta<1$. Now  Theorem \ref{thm6} follows from Corollary \ref{cor1}.
\end{proof}
\begin{remark}
	In Theorem \ref{thm6}, 
	the arithmetic  condition, namely Diophantine condition, on frequencies  $\alpha$ is necessary. For some non-Diophantine $\alpha$, 
	$	\langle|\tilde{X}_{H_x}|_{\phi}^p	\rangle(T)$ could have the power law lower bound ($	\langle|\tilde{X}_{H_x}|_{\phi}^p	\rangle(T)\geq T^{\gamma}$ for some $\gamma>0$) \cite{jz22}. 
\end{remark}
	\subsection{Shifts: $b=1$, arbitrary $d$}
	
	%	Let $v $ be analytic on  $\T$.
	Let
	\begin{equation*}
	f^n(x)=x+n\alpha=x+n_1\alpha_1+n_2\alpha_2+\cdots +n_d\alpha_d\mod\Z,
	\end{equation*}
	where $n=(n_1,n_2,\cdots,n_d)\in \Z^d$ and $x\in \T$.
	Let $H_x$ on $\ell^2(\Z^d)$ be given by
	\begin{equation}\label{opapp4}
	H_x=\lambda^{-1}S+ v(f^n(x))\delta_{nn^\prime}=\lambda^{-1}S+ v(x+n_1\alpha_1+n_2\alpha_2+\cdots +n_d\alpha_d)\delta_{nn^\prime},
	\end{equation}
	where $v$ is a non-constant real analytic function  on $\T$.
	
	\begin{theorem}\label{thm5}
	Let $\alpha\in {\rm DC}(\kappa,\tau)$ and  $H_x$ be given by \eqref{opapp4}. Then there exists $\lambda_0=\lambda_0(\kappa,\tau,C_1,c_1,v)$ such that
	for any $\lambda>\lambda_0$ the following holds. For   any $\varepsilon>0$ and any $\phi$ with compact support,  there exists $T_0>0$ (depending on $p, \phi,S, v,\alpha$ and $\varepsilon$) such that for any $T\geq T_0$
	\begin{equation}  
	\langle|\tilde{X}_{H_x}|_{\phi}^p\rangle(T)	\leq  (\ln T)^{  p+\varepsilon}
	\end{equation}
	and  for any $t\geq T_0$
	\begin{equation}  
	\langle|{X}_{H_x}|_{\phi}^p\rangle(t)	\leq  (\ln t)^{p+\varepsilon}.
	\end{equation}
\end{theorem}
\begin{proof}
By \cite[Theorem 3.11]{liuapde},  $H_x$ satisfies LDT. 
We can  assume that $ X_{N}$ (given by \eqref{gdec3}) is a semi-algebraic set with degree less than $N^C$. 
Then  \eqref{gsublinear} holds for any $\delta<1$. Now  Theorem \ref{thm5} follows from Corollary \ref{cor1}.
\end{proof}

	\subsection{Shifts: $d=b=2$}
	Assume $v$ is real analytic on $\T^2$.
	Let
	\begin{equation*}
	f^n(x)=(x_1+n_1\alpha_1,x_2+n_2\alpha_2)\mod\Z^2,
	\end{equation*}
	where $n=(n_1,n_2)\in \Z^2$, $\alpha=(\alpha_1,\alpha_2)\in \R^2$ and $x=(x_1,x_2)\in \T^2$.
	Let $H_x$ on $\ell^2(\Z^2)$ be given by
	\begin{equation}\label{opapp6}
	H_x=\lambda^{-1}S+v(f^n(x))\delta_{nn^\prime}=\lambda^{-1}S+
	v(x_1+n_1\alpha_1,x_2+n_2\alpha_2)\delta_{nn^\prime}.
	\end{equation}

	\begin{theorem}\label{thmapp3}
		Let $H_x$ be given by \eqref{opapp6}.
		Suppose $v$ is non-constant on any line segment contained in $[0,1)^2$, $\alpha_1\in {\rm DC}(\kappa,\tau)$ and $\alpha_2\in {\rm DC}(\kappa,\tau)$ with $1\leq  \kappa<\frac{13}{12}$.
		Then there exists $\lambda_0=\lambda_0(\kappa,\tau,c_1,C_1,v)$ such that
		for any $\lambda>\lambda_0$ the following holds. For  any $\phi$ with compact support and any $\varepsilon>0$,  there exists $T_0>0$ (depending on $p,\phi,S, v,\alpha$ and $\varepsilon$) such that for any $T\geq T_0$
		\begin{equation}  
		\langle|\tilde{X}_{H_x}|_{\phi}^p\rangle(T)	\leq  (\ln T)^{  \frac{4p}{13-12\kappa}+\varepsilon}
		\end{equation}
		and  for any $t\geq T_0$
		\begin{equation}  \label{g22}
		\langle|{X}_{H_x}|_{\phi}^p\rangle(t)	\leq  (\ln t)^{\frac{4p}{13-12\kappa}+\varepsilon}.
		\end{equation}
		 
	\end{theorem}
 
 \begin{proof}
 Recall that under the assumption in Theorem \ref{thmapp3}, $H_x$ satisfies LDT (see \cite[Theorem 3.20]{liuapde}).
 	By  \cite[Theorem 5.1]{bk19} (also item 1 of Remark 11  in \cite{liuapde}) , \eqref{gsublinear} holds for any $ \delta $ with $0<\delta<\frac{13}{4}-3\kappa$.
 \end{proof}
 
 \begin{remark}\label{re5}
\begin{enumerate}
	\item The LDT of the operator \eqref{opapp6} was established by the author in \cite{liuapde}, which builds on an earlier work of Bourgain–Kachkovskiy \cite{bk19}.
	\item A larger  power law logarithmic   bound  $ (\frac{4}{13-12\kappa})^2p$ was established by Shamis-Sodin  \cite{SS21}.
\end{enumerate}

 \end{remark}

\section*{Acknowledgments}
%The authors are very grateful to the anonymous referees for their knowledgeable reports, which helped us to improve our manuscript.
 
This research was 
supported by   NSF   DMS-2000345  and DMS-2052572.
	\section*{Conflict of Interest Statement }
The author  declares no conflicts of interest.
\section*{Data Availability Statement}
Data sharing is not applicable to this article as no new data were created or analyzed in this study.

% ------------------------------------------------------------------------

\begin{thebibliography}{10}
	
	\bibitem{a94}
	M.~Aizenman.
	\newblock Localization at weak disorder: some elementary bounds.
	\newblock volume~6, pages 1163--1182. 1994.
	\newblock Special issue dedicated to Elliott H. Lieb.
	
	\bibitem{aenss}
	M.~Aizenman, A.~Elgart, S.~Naboko, J.~H. Schenker, and G.~Stolz.
	\newblock Moment analysis for localization in random {S}chr\"{o}dinger
	operators.
	\newblock {\em Invent. Math.}, 163(2):343--413, 2006.
	
	\bibitem{am}
	M.~Aizenman and S.~Molchanov.
	\newblock Localization at large disorder and at extreme energies: an elementary
	derivation.
	\newblock {\em Comm. Math. Phys.}, 157(2):245--278, 1993.
	
	\bibitem{ayz}
	A.~Avila, J.~You, and Q.~Zhou.
	\newblock Sharp phase transitions for the almost {M}athieu operator.
	\newblock {\em Duke Math. J.}, 166(14):2697--2718, 2017.
	
	\bibitem{as82}
	J.~Avron and B.~Simon.
	\newblock Singular continuous spectrum for a class of almost periodic {J}acobi
	matrices.
	\newblock {\em Bull. Amer. Math. Soc. (N.S.)}, 6(1):81--85, 1982.
	
	\bibitem{bbook}
	J.~Bourgain.
	\newblock {\em Green's function estimates for lattice {S}chr\"{o}dinger
		operators and applications}, volume 158 of {\em Annals of Mathematics
		Studies}.
	\newblock Princeton University Press, Princeton, NJ, 2005.
	
	\bibitem{boujam}
	J.~Bourgain.
	\newblock Positivity and continuity of the {L}yapounov exponent for shifts on
	{$\Bbb T^d$} with arbitrary frequency vector and real analytic potential.
	\newblock {\em J. Anal. Math.}, 96:313--355, 2005.
	
	\bibitem{bou07}
	J.~Bourgain.
	\newblock Anderson localization for quasi-periodic lattice {S}chr\"{o}dinger
	operators on {$\Bbb Z^d$}, {$d$} arbitrary.
	\newblock {\em Geom. Funct. Anal.}, 17(3):682--706, 2007.
	
	\bibitem{bg}
	J.~Bourgain and M.~Goldstein.
	\newblock On nonperturbative localization with quasi-periodic potential.
	\newblock {\em Ann. of Math. (2)}, 152(3):835--879, 2000.
	
	\bibitem{bgs01}
	J.~Bourgain, M.~Goldstein, and W.~Schlag.
	\newblock Anderson localization for {S}chr\"{o}dinger operators on {$\Bbb Z$}
	with potentials given by the skew-shift.
	\newblock {\em Comm. Math. Phys.}, 220(3):583--621, 2001.
	
	\bibitem{BGS}
	J.~Bourgain, M.~Goldstein, and W.~Schlag.
	\newblock Anderson localization for {S}chr\"{o}dinger operators on {$\bold
		Z^2$} with quasi-periodic potential.
	\newblock {\em Acta Math.}, 188(1):41--86, 2002.
	
	\bibitem{bk19}
	J.~Bourgain and I.~Kachkovskiy.
	\newblock Anderson localization for two interacting quasiperiodic particles.
	\newblock {\em Geom. Funct. Anal.}, 29(1):3--43, 2019.
	
	\bibitem{bk05}
	J.~Bourgain and C.~E. Kenig.
	\newblock On localization in the continuous {A}nderson-{B}ernoulli model in
	higher dimension.
	\newblock {\em Invent. Math.}, 161(2):389--426, 2005.
	
	\bibitem{df}
	V.~Bucaj, D.~Damanik, J.~Fillman, V.~Gerbuz, T.~VandenBoom, F.~Wang, and
	Z.~Zhang.
	\newblock Localization for the one-dimensional {A}nderson model via positivity
	and large deviations for the {L}yapunov exponent.
	\newblock {\em Trans. Amer. Math. Soc.}, 372(5):3619--3667, 2019.
	
	\bibitem{ckm87}
	R.~Carmona, A.~Klein, and F.~Martinelli.
	\newblock Anderson localization for {B}ernoulli and other singular potentials.
	\newblock {\em Comm. Math. Phys.}, 108(1):41--66, 1987.
	
	\bibitem{CFKS87}
	H.~L. Cycon, R.~G. Froese, W.~Kirsch, and B.~Simon.
	\newblock {\em Schr\"{o}dinger operators with application to quantum mechanics
		and global geometry}.
	\newblock Texts and Monographs in Physics. Springer-Verlag, Berlin, study
	edition, 1987.
	
	\bibitem{ds01}
	D.~Damanik and P.~Stollmann.
	\newblock Multi-scale analysis implies strong dynamical localization.
	\newblock {\em Geom. Funct. Anal.}, 11(1):11--29, 2001.
	
	\bibitem{DK07}
	D.~Damanik and S.~Tcheremchantsev.
	\newblock Upper bounds in quantum dynamics.
	\newblock {\em J. Amer. Math. Soc.}, 20(3):799--827, 2007.
	
	\bibitem{DKJFA}
	D.~Damanik and S.~Tcheremchantsev.
	\newblock Quantum dynamics via complex analysis methods: general upper bounds
	without time-averaging and tight lower bounds for the strongly coupled
	{F}ibonacci {H}amiltonian.
	\newblock {\em J. Funct. Anal.}, 255(10):2872--2887, 2008.
	
	\bibitem{ds20}
	J.~Ding and C.~K. Smart.
	\newblock Localization near the edge for the {A}nderson {B}ernoulli model on
	the two dimensional lattice.
	\newblock {\em Invent. Math.}, 219(2):467--506, 2020.
	
	\bibitem{dt97}
	M.~Drmota and R.~F. Tichy.
	\newblock {\em Sequences, discrepancies and applications}, volume 1651 of {\em
		Lecture Notes in Mathematics}.
	\newblock Springer-Verlag, Berlin, 1997.
	
	\bibitem{DKJEMS}
	P.~Duarte and S.~Klein.
	\newblock Continuity, positivity and simplicity of the {L}yapunov exponents for
	quasi-periodic cocycles.
	\newblock {\em J. Eur. Math. Soc. (JEMS)}, 21(7):2051--2106, 2019.
	
	\bibitem{ek14}
	A.~Elgart and A.~Klein.
	\newblock Ground state energy of trimmed discrete {S}chr\"{o}dinger operators
	and localization for trimmed {A}nderson models.
	\newblock {\em J. Spectr. Theory}, 4(2):391--413, 2014.
	
	\bibitem{ek16}
	A.~Elgart and A.~Klein.
	\newblock An eigensystem approach to {A}nderson localization.
	\newblock {\em J. Funct. Anal.}, 271(12):3465--3512, 2016.
	
	\bibitem{ek19}
	A.~Elgart and A.~Klein.
	\newblock Eigensystem multiscale analysis for {A}nderson localization in energy
	intervals.
	\newblock {\em J. Spectr. Theory}, 9(2):711--765, 2019.
	
	\bibitem{ek20}
	A.~Elgart and A.~Klein.
	\newblock Eigensystem multiscale analysis for the {A}nderson model via the
	{W}egner estimate.
	\newblock {\em Ann. Henri Poincar\'{e}}, 21(7):2301--2326, 2020.
	
	\bibitem{eks18}
	A.~Elgart, A.~Klein, and G.~Stolz.
	\newblock Manifestations of dynamical localization in the disordered {XXZ} spin
	chain.
	\newblock {\em Comm. Math. Phys.}, 361(3):1083--1113, 2018.
	
	\bibitem{eks182}
	A.~Elgart, A.~Klein, and G.~Stolz.
	\newblock Many-body localization in the droplet spectrum of the random {XXZ}
	quantum spin chain.
	\newblock {\em J. Funct. Anal.}, 275(1):211--258, 2018.
	
	\bibitem{fmss}
	J.~Fr\"{o}hlich, F.~Martinelli, E.~Scoppola, and T.~Spencer.
	\newblock Constructive proof of localization in the {A}nderson tight binding
	model.
	\newblock {\em Comm. Math. Phys.}, 101(1):21--46, 1985.
	
	\bibitem{fs}
	J.~Fr\"{o}hlich and T.~Spencer.
	\newblock Absence of diffusion in the {A}nderson tight binding model for large
	disorder or low energy.
	\newblock {\em Comm. Math. Phys.}, 88(2):151--184, 1983.
	
	\bibitem{gd}
	F.~Germinet and S.~de~Bi\`evre.
	\newblock Localisation dynamique et op\'{e}rateurs de {S}chr\"{o}dinger
	al\'{e}atoires.
	\newblock {\em C. R. Acad. Sci. Paris S\'{e}r. I Math.}, 326(2):261--264, 1998.
	
	\bibitem{ghk}
	F.~Germinet, P.~D. Hislop, and A.~Klein.
	\newblock Localization for {S}chr\"{o}dinger operators with {P}oisson random
	potential.
	\newblock {\em J. Eur. Math. Soc. (JEMS)}, 9(3):577--607, 2007.
	
	\bibitem{GKT04}
	F.~Germinet, A.~Kiselev, and S.~Tcheremchantsev.
	\newblock Transfer matrices and transport for {S}chr\"{o}dinger operators.
	\newblock {\em Ann. Inst. Fourier (Grenoble)}, 54(3):787--830, 2004.
	
	\bibitem{gk01}
	F.~Germinet and A.~Klein.
	\newblock Bootstrap multiscale analysis and localization in random media.
	\newblock {\em Comm. Math. Phys.}, 222(2):415--448, 2001.
	
	\bibitem{gk03}
	F.~Germinet and A.~Klein.
	\newblock Explicit finite volume criteria for localization in continuous random
	media and applications.
	\newblock {\em Geom. Funct. Anal.}, 13(6):1201--1238, 2003.
	
	\bibitem{gk04}
	F.~Germinet and A.~Klein.
	\newblock A characterization of the {A}nderson metal-insulator transport
	transition.
	\newblock {\em Duke Math. J.}, 124(2):309--350, 2004.
	
	\bibitem{gk13}
	F.~Germinet and A.~Klein.
	\newblock A comprehensive proof of localization for continuous {A}nderson
	models with singular random potentials.
	\newblock {\em J. Eur. Math. Soc. (JEMS)}, 15(1):53--143, 2013.
	
	\bibitem{gks}
	F.~Germinet, A.~Klein, and J.~H. Schenker.
	\newblock Dynamical delocalization in random {L}andau {H}amiltonians.
	\newblock {\em Ann. of Math. (2)}, 166(1):215--244, 2007.
	
	\bibitem{gks09}
	F.~Germinet, A.~Klein, and J.~H. Schenker.
	\newblock Quantization of the {H}all conductance and delocalization in ergodic
	{L}andau {H}amiltonians.
	\newblock {\em Rev. Math. Phys.}, 21(8):1045--1080, 2009.
	
	\bibitem{gmp}
	I.~J. Goldsheid, S.~A. Mol\v{c}anov, and L.~A. Pastur.
	\newblock A random homogeneous {S}chr\"{o}dinger operator has a pure point
	spectrum.
	\newblock {\em Funkcional. Anal. i Prilo\v{z}en.}, 11(1):1--10, 96, 1977.
	
	\bibitem{GS01}
	M.~Goldstein and W.~Schlag.
	\newblock H\"{o}lder continuity of the integrated density of states for
	quasi-periodic {S}chr\"{o}dinger equations and averages of shifts of
	subharmonic functions.
	\newblock {\em Ann. of Math. (2)}, 154(1):155--203, 2001.
	
	\bibitem{gordon1976}
	A.~Y. Gordon.
	\newblock The point spectrum of the one-dimensional {S}chr\"odinger operator.
	\newblock {\em Uspekhi Matematicheskikh Nauk}, 31(4):257--258, 1976.
	
	\bibitem{HJ19}
	R.~Han and S.~Jitomirskaya.
	\newblock Quantum dynamical bounds for ergodic potentials with underlying
	dynamics of zero topological entropy.
	\newblock {\em Anal. PDE}, 12(4):867--902, 2019.
	
	\bibitem{jl222}
	S.~Jitomirskaya and W.~Liu.
	\newblock Universal reflective-hierarchical structure of quasiperiodic
	eigenfunctions and sharp spectral transition in phase.
	\newblock {\em J. Eur. Math. Soc. (JEMS) to appear}.
	
	\bibitem{jl118}
	S.~Jitomirskaya and W.~Liu.
	\newblock Universal hierarchical structure of quasiperiodic eigenfunctions.
	\newblock {\em Ann. of Math. (2)}, 187(3):721--776, 2018.
	
	\bibitem{JL21}
	S.~Jitomirskaya and W.~Liu.
	\newblock Upper bounds on transport exponents for long-range operators.
	\newblock {\em J. Math. Phys.}, 62(7):Paper No. 073506, 9, 2021.
	
	\bibitem{jls20}
	S.~Jitomirskaya, W.~Liu, and Y.~Shi.
	\newblock Anderson localization for multi-frequency quasi-periodic operators on
	{${\Bbb Z}^D$}.
	\newblock {\em Geom. Funct. Anal.}, 30(2):457--481, 2020.
	
	\bibitem{JM17}
	S.~Jitomirskaya and R.~Mavi.
	\newblock Dynamical bounds for quasiperiodic {S}chr\"{o}dinger operators with
	rough potentials.
	\newblock {\em Int. Math. Res. Not. IMRN}, (1):96--120, 2017.
	
	\bibitem{JP22}
	S.~Jitomirskaya and M.~Powell.
	\newblock Logarithmic quantum dynamical bounds for arithmetically defined
	ergodic {S}chr{\"o}dinger operators with smooth potentials.
	\newblock {\em Analysis at Large: Dedicated to the Life and Work of Jean
		Bourgain}, pages 173--201, 2022.
	
	\bibitem{js07}
	S.~Jitomirskaya and H.~Schulz-Baldes.
	\newblock Upper bounds on wavepacket spreading for random {J}acobi matrices.
	\newblock {\em Comm. Math. Phys.}, 273(3):601--618, 2007.
	
	\bibitem{jss}
	S.~Jitomirskaya, H.~Schulz-Baldes, and G.~Stolz.
	\newblock Delocalization in random polymer models.
	\newblock {\em Comm. Math. Phys.}, 233(1):27--48, 2003.
	
	\bibitem{js94}
	S.~Jitomirskaya and B.~Simon.
	\newblock Operators with singular continuous spectrum. {III}. {A}lmost periodic
	{S}chr\"{o}dinger operators.
	\newblock {\em Comm. Math. Phys.}, 165(1):201--205, 1994.
	
	\bibitem{jz22}
	S.~Jitomirskaya and S.~Zhang.
	\newblock Quantitative continuity of singular continuous spectral measures and
	arithmetic criteria for quasiperiodic {S}chr\"{o}dinger operators.
	\newblock {\em J. Eur. Math. Soc. (JEMS)}, 24(5):1723--1767, 2022.
	
	\bibitem{jxz19}
	S.~Jitomirskaya and X.~Zhu.
	\newblock Large deviations of the {L}yapunov exponent and localization for the
	1{D} {A}nderson model.
	\newblock {\em Comm. Math. Phys.}, 370(1):311--324, 2019.
	
	\bibitem{JL20}
	S.~Y. Jitomirskaya and Y.~Last.
	\newblock Power law subordinacy and singular spectra. {II}. {L}ine operators.
	\newblock {\em Comm. Math. Phys.}, 211(3):643--658, 2000.
	
	\bibitem{kn13}
	A.~Klein and S.~T. Nguyen.
	\newblock The bootstrap multiscale analysis of the multi-particle {A}nderson
	model.
	\newblock {\em J. Stat. Phys.}, 151(5):938--973, 2013.
	
	\bibitem{kn15}
	A.~Klein and S.~T. Nguyen.
	\newblock Bootstrap multiscale analysis and localization for multi-particle
	continuous {A}nderson {H}amiltonians.
	\newblock {\em J. Spectr. Theory}, 5(2):399--444, 2015.
	
	\bibitem{knr}
	A.~Klein, S.~T. Nguyen, and C.~Rojas-Molina.
	\newblock Characterization of the metal-insulator transport transition for the
	two-particle {A}nderson model.
	\newblock {\em Ann. Henri Poincar\'{e}}, 18(7):2327--2365, 2017.
	
	\bibitem{ks11}
	A.~Klein and C.~Sadel.
	\newblock Ballistic behavior for random {S}chr\"{o}dinger operators on the
	{B}ethe strip.
	\newblock {\em J. Spectr. Theory}, 1(4):409--442, 2011.
	
	\bibitem{ks12}
	A.~Klein and C.~Sadel.
	\newblock Absolutely continuous spectrum for random {S}chr\"{o}dinger operators
	on the {B}ethe strip.
	\newblock {\em Math. Nachr.}, 285(1):5--26, 2012.
	
	\bibitem{kt18}
	A.~Klein and C.~S.~S. Tsang.
	\newblock Eigensystem bootstrap multiscale analysis for the {A}nderson model.
	\newblock {\em J. Spectr. Theory}, 8(3):1149--1197, 2018.
	
	\bibitem{ks80}
	H.~Kunz and B.~Souillard.
	\newblock Sur le spectre des op\'{e}rateurs aux diff\'{e}rences finies
	al\'{e}atoires.
	\newblock {\em Comm. Math. Phys.}, 78(2):201--246, 1980/81.
	
	\bibitem{l96}
	Y.~Last.
	\newblock Quantum dynamics and decompositions of singular continuous spectra.
	\newblock {\em J. Funct. Anal.}, 142(2):406--445, 1996.
	
	\bibitem{lz22}
	L.~Li and L.~Zhang.
	\newblock Anderson-{B}ernoulli localization on the three-dimensional lattice
	and discrete unique continuation principle.
	\newblock {\em Duke Math. J.}, 171(2):327--415, 2022.
	
	\bibitem{liuapde}
	W.~Liu.
	\newblock Quantitative inductive estimates for {G}reen's functions of
	non-self-adjoint matrices.
	\newblock {\em Analysis and PDE to appear}.
	
	\bibitem{p22}
	M.~Powell.
	\newblock Continuity of the lyapunov exponent for analytic multi-frequency
	quasiperiodic cocycles.
	\newblock {\em arXiv preprint arXiv:2210.09285}, 2022.
	
	\bibitem{r12}
	C.~Rojas-Molina.
	\newblock Characterization of the {A}nderson metal-insulator transition for non
	ergodic operators and application.
	\newblock {\em Ann. Henri Poincar\'{e}}, 13(7):1575--1611, 2012.
	
	\bibitem{sar}
	P.~Sarnak.
	\newblock Spectral behavior of quasiperiodic potentials.
	\newblock {\em Comm. Math. Phys.}, 84(3):377--401, 1982.
	
	\bibitem{Sc22}
	W.~Schlag.
	\newblock An introduction to multiscale techniques in the theory of {A}nderson
	localization, {P}art {I}.
	\newblock {\em Nonlinear Anal.}, 220:Paper No. 112869, 55, 2022.
	
	\bibitem{SS21}
	M.~Shamis and S.~Sodin.
	\newblock Upper bounds on quantum dynamics in arbitrary dimension.
	\newblock {\em arXiv preprint arXiv:2111.10902}, 2021.
	
	\bibitem{dk89}
	H.~von Dreifus and A.~Klein.
	\newblock A new proof of localization in the {A}nderson tight binding model.
	\newblock {\em Comm. Math. Phys.}, 124(2):285--299, 1989.
	
\end{thebibliography}
\end{document}